%% file: arxiv.tex
\documentclass[11pt,letterpaper]{article}
\usepackage[lmargin=1.0in,rmargin=1.0in,bottom=1.0in,top=1.0in,twoside=False]{geometry}

\usepackage{fullpage,amssymb,amsmath,amsthm}
\usepackage{graphicx}
\usepackage{enumerate}
\usepackage[T1]{fontenc}
\usepackage{xcolor}
\usepackage{amsfonts}
\usepackage{comment}
\usepackage[english]{babel}
\usepackage{libertine}
\usepackage[absolute]{textpos}
\usepackage{enumitem}
\usepackage{mathtools}
\usepackage{mathrsfs}
\usepackage{todonotes}

\theoremstyle{definition}
\newtheorem{definition}{Definition}[section]
\theoremstyle{plain}
\newtheorem{theorem}[definition]{Theorem}
\newtheorem{lemma}[definition]{Lemma}

\newtheorem{corollary}[definition]{Corollary}

\definecolor{MyBlue}{RGB}{50,100,200}
\usepackage[pdftex]{hyperref}
\hypersetup{%
  colorlinks=true,
  linkcolor=MyBlue,
  citecolor=MyBlue,
  urlcolor=MyBlue
}

\usepackage[T1]{fontenc}

\usepackage{enumitem}

\usepackage[absolute]{textpos}
\usepackage{microtype}
\usepackage{amsfonts}
\usepackage{comment}
\usepackage{mathrsfs}
\usepackage{todonotes}
\usepackage{complexity}
\usepackage{authblk}

\begin{document}

\title{Simple and tight complexity lower bounds for solving Rabin games\thanks{This
work is a part of projects CUTACOMBS (Ma. Pilipczuk) and BOBR (Mi. Pilipczuk) 
%, and VAMOS (K. S. Thejaswini)
that have received funding from the European Research Council (ERC) under the European Union’s Horizon 2020 research
and innovation programme, grant agreements No 714704 and 948057
%, and 101020093, 
respectively.
Ma. Pilipczuk is also partially supported by Polish National Science Centre SONATA BIS-12 grant number 2022/46/E/ST6/00143.}}
\author[1]{Antonio Casares}
\author[2]{Marcin Pilipczuk}
\author[2]{Michał Pilipczuk}
\author[3]{U\'everton S. Souza}
\author[4,5]{{K.~S.~Thejaswini}}
\affil[1]{LaBRI, Université de Bordeaux, France}
\affil[2]{University of Warsaw, Poland}
\affil[3]{Universidade Federal Fluminense, Niter\'{o}i, Brazil}
\affil[4]{University of Warwick, United Kingdom}
\affil[5]{Institute of Science and Technology, Austria}
\setcounter{Maxaffil}{0}
\renewcommand\Affilfont{\itshape\small}

\date{}

\maketitle
\input{macros}

\begin{abstract}
We give a simple proof that assuming the Exponential Time Hypothesis (ETH), determining the winner of a Rabin game cannot be done in time $2^{o(k \log k)} \cdot n^{\Oh(1)}$, where $k$ is the number of pairs of vertex subsets involved in the winning condition and $n$ is the vertex count of the game graph. While this result follows from the lower bounds provided by Calude et al [SIAM J. Comp. 2022], our reduction is considerably simpler and arguably provides more insight into the complexity of the problem. In fact, the analogous lower bounds discussed by Calude et al, for solving Muller games and multidimensional parity games, follow as simple corollaries of our approach. Our reduction also highlights the usefulness of a certain pivot problem --- {\sc{Permutation SAT}} --- which may be of independent interest.
\end{abstract}

\begin{textblock}{20}(0, 12.7)
\includegraphics[width=35px]{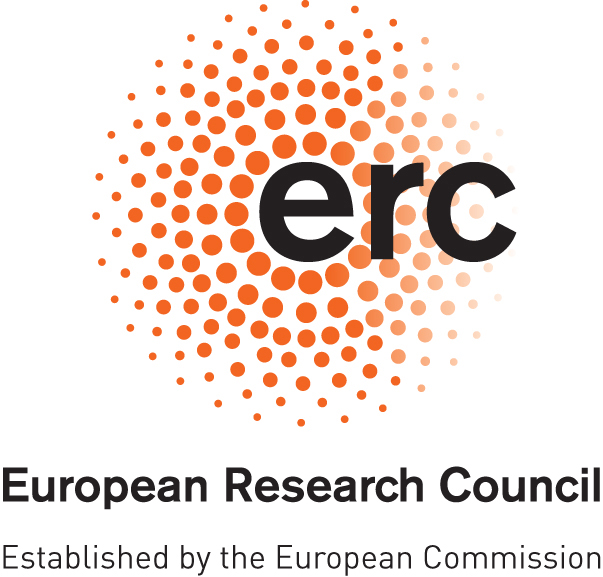}%
\end{textblock}
\begin{textblock}{20}(0, 13.6)
\includegraphics[width=35px]{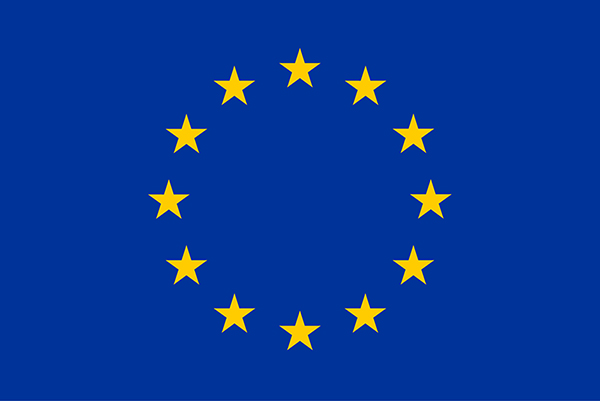}%
\end{textblock}

\input{0Intro}

\section{Preliminaries on games}

For a positive integer $p$, we denote $[p]\coloneqq \{1,\ldots,p\}$.

Rabin and Muller games are turn-based two-player games played on an \emph{arena} that is a directed graph $D = (V,E)$ together with a partition of the vertices into those owned by player \Steven{} and those owned by player \Audrey. 
A token is initially placed on a designated starting vertex $u_1$. In each consecutive turn, the owner of the vertex bearing the token moves the token along an edge of $D$. Thus, the players jointly form an infinite sequence of vertices in consecutive turns, referred to as a \emph{play}. An \emph{objective} is a representation of a subset of the set of all possible plays. We will consider three different objectives discussed below.

\paragraph{Muller objectives.} 
In a Muller game, each vertex is labelled with a subset of colours from $[k]$ via a mapping $c\colon V\to 2^{[k]}$, where $V$ is the set of vertices of the arena $D$. The Muller objective is specified by a family of subsets of colours $\Fc\subseteq 2^{[k]}$. A play $\rho$ is winning for \Steven{} if the set of colours visited infinitely often, belongs to $\Fc$, that is, if
\[ \bigcup_{v\in\infSeq(\rho)} c(v) \; \in  \Fc,\]
where $\infSeq(\rho)$ is the set of vertices appearing infinitely often in the play.

%The arena to describe a $k$-Muller game is has each vertex equipped with subsets of integers from  $\{1,\dots, k\}$. The Muller objective is captured by such a colouring along with a set of subsets of colours $\Fc\subseteq 2^{[k]}$. 
%A play satisfies the $k$-Muller objective for \Steven~if the union of these subsets of indices seen infinitely often along such a play is exactly an element of $\Fc$. If not, the play is winning for \Audrey.

\paragraph{Rabin objective.}
%Rabin objectives are specific kinds of Muller objectives. 
A Rabin objective of \emph{degree} $k$ consists of $k$ pairs of vertex subsets $(G_1,B_1),\ldots,(G_k,B_k)$; $G_i$ is said to be the \emph{good} subset for index $i$, and $B_i$ is the \emph{bad} subset. A play $\rho$ is winning for \Steven~if there exists an index $i\in \{1,\ldots,k\}$ such that $\rho$ visits $G_i$ infinitely often and $B_i$ only a finite number of times. 

Rabin objectives of degree $k$ can be encoded as a Muller objective using $2k$ colours. Indeed, for each $0\leq i<k$, we associate $2i$ with the subset $G_i$ and $2i+1$ with the subset $B_i$. We define $c\colon V\to 2^{[2k]}$ and~$\Fc$~as:
\[ c(v) = \{2i \colon v\in G_i\} \cup \{2i+1 \colon  v\in B_i\} \;\; \text{ and } \;\; \Fc = \{C\subseteq [2k] \colon \exists i\; 2i\in C \text{ and } 2i+1 \notin C\}. \]

\paragraph{Generalised parity objective.}
Generalised parity games were first considered in the work of Chaterjee, Henzinger, and Piterman~\cite{CHP07}. In a $d$-dimensional $k$-parity condition, each vertex is labelled with a $d$-dimensional vector of integers from $\{1,\ldots,k\}$.
An infinite play satisfies this objective for \Steven~if and only if there is some coordinate such that the highest number that occurs infinitely often at this coordinate is even.  \Audrey~wins otherwise.

%\antonioBox{Please, check this} 
%These games are at least as hard as Rabin games, as a $d$-dimensional $k$-parity game can be reduced to a Rabin game of degree $dk/2$ (as shown by \cite{CHP07}). 
These games are inter-reducible with Rabin games, as shown by \cite{CHP07}. For one direction, since a $d$-dimensional $k$-parity objective is a disjunction of $d$ distinct parity objectives, and  each parity objective can be expressed as a Rabin objective of degree $\lceil k/2\rceil$, the $d$-dimensional $k$-parity objective can therefore similarly be transformed into a Rabin objective of degree $d\lceil k/2\rceil$, with $\lceil k/2\rceil$ Rabin pairs for \emph{each} of the $d$ parity objectives. Conversely, a Rabin objective with $d$ pairs can be represented as a $d$-dimensional $3$-parity objective. Indeed, we use each pair $(G_i,B_i)$ to define the component $p_i$ that assigns colour 3 to $v$ when $v\in B_i$, colour 2 if $v\in G_i\setminus B_i$ and 1 otherwise.

% Assume that a $d$-dimensional $k$-parity objective is given by $d$ functions $p_i\colon V \to \{1,\dots, k\}$ assigning the $i^{th}$ component of the vectors. To define an equivalent Rabin objective, we consider pairs of subsets $(G_{i,j},B_{i,j})$ for $1\leq i \leq d$ and $0\leq j \leq \lfloor k/2 \rfloor$. We define 
% $$G_{i,j} = \{v\mid p_i(v) = 2j \} \qquad\text{ and } \qquad B_{i,j} = \{v\mid p_i(v) > 2j\}\}$$
% $v\in G_{i,j}$ if $p_i(v) = 2j$ and $v\in B_{i,j}$ if $p_i(v) = l$ for some odd $l< 2j$.
% Conversely, a Rabin objective with $d$ pairs can be represented as a $d$-dimensional $3$-parity objective. Indeed, we use each pair $(G_i,B_i)$ to define the component $p_i$ given by $p_i(v) = 3$ if $v\in B_i$, $p_i(v) = 2$ if $v\in G_i\setminus B_i$ and $p_i(v) = 1$ otherwise.

Calude et al.~\cite{CJKLS22} showed that generalised parity games cannot be solved in time $2^{o(k\log k)}\cdot n^{\Oh(1)}$ assuming the ETH, even when the dimensions is $d=2$. %Although the definition considered by Calude et al. is the dual of what we defined here (They consider \Audrey's objective  if \emph{for all} dimensions the highest priority is even), 

\paragraph{Strategies and winners.} For a given game, with any of the objectives discussed above,
 a \emph{strategy}  of \Steven~is a function from the set of plays ending at a \Steven~vertex to the set of vertices. 
 A play $v_0,v_1,\dots,v_i,\dots$ is said to respect this strategy if for every vertex $v_i$ which belongs to \Steven, the vertex $v_{i+1}$ is the one proposed by the strategy on the finite prefix of this play ending at $v_i$. 
For a fixed objective, a game is said to be winning for \Steven~if he has a strategy such that plays respecting this strategy satisfy the~objective.

\paragraph{Positional strategies.}  We say that a strategy (for \Steven{}) is \emph{positional} (or \emph{memoryless}) if it can be represented by a function assigning an outgoing edge to each vertex owned by Steven. That is, a positional strategy always makes the same decision over the same vertex, and this decision depends only on the current vertex and not on the history of the play. It is well known that Rabin games are positional for Steven in the following sense.

\begin{lemma}[\cite{EJ88,EJ99}]
Rabin games are positional for \Steven{}. That is, if \Steven{} wins a Rabin game, then he has a positional winning strategy.
\end{lemma}

 \paragraph{Exponential Time Hypothesis.} The Exponential Time Hypothesis is a complexity assumption introduced by Impagliazzo, Paturi and Zane~\cite{ImpagliazzoPZ01} that postulates the following: there exists $\delta>0$ such that the \SATThree{} problem cannot be solved in time $\Oh(2^{\delta n})$, where $n$ is the number of variables of the input formula. We refer the reader to~\cite[Chapter~14]{platypus} for an introduction to the applications of ETH for lower bounds within parameterized complexity.

\section{Permutation SAT}\label{sec:permSAT}
\input{2PermSAT.tex}

\section{Lower bound for Rabin games}
\input{3RabinGames.tex}

\paragraph*{Acknowledgements.} A large part of the results presented in this paper were obtained during Autobóz 2023, an annual research camp on automata theory. The authors thank the organisers and participants of Autobóz for creating a wonderful research atmosphere.

\bibliographystyle{alpha}
\bibliography{games}

\end{document}

%% file: macros.tex
\renewcommand{\leq}{\leqslant}
\renewcommand{\geq}{\geqslant}
\renewcommand{\le}{\leqslant}
\renewcommand{\ge}{\geqslant}

\newcommand{\Bc}{\mathcal{B}}
\newcommand{\Cc}{\mathcal{C}}
\newcommand{\Dc}{\mathcal{D}}
\newcommand{\Fc}{\mathcal{F}}
\newcommand{\Gc}{G}
\newcommand{\Hc}{\mathcal{H}}
\newcommand{\Lc}{\mathcal{L}}
\newcommand{\Mc}{\mathcal{M}}
\newcommand{\Rc}{\mathcal{R}}
\newcommand{\Tc}{\mathcal{T}}
\newcommand{\Uc}{\mathcal{U}}
\newcommand{\Vc}{\mathcal{V}}
\newcommand{\Wc}{\mathcal{W}}
\newcommand{\Even}{\mathrm{Even}}
\newcommand{\Odd}{\mathrm{Odd}}
\newcommand{\Reg}[2]{\Rc^{#1}\!\left(#2\right)}
\newcommand{\Def}[2]{\Dc^{#1}\!\left(#2\right)}
\newcommand{\floor}[1]{\left\lfloor #1 \right\rfloor}
\newcommand{\ceil}[1]{\left\lceil #1 \right\rceil}
\newcommand{\seq}[1]{\left\langle #1 \right\rangle}
\newcommand{\tpl}[1]{\left( #1 \right)}
\newcommand{\eset}[1]{\left\{\, #1 \,\right\}}
\newcommand{\Nats}{\mathbb{N}}
\newcommand{\mymax}{\mathrm{max}}
\newcommand{\infSeq}{\mathrm{Inf}}

\newcommand{\Steven}{Steven}
\newcommand{\Audrey}{Audrey}
\newcommand{\permsat}{4\text{-}\textsc{Permutation SAT}}
\newcommand{\SATThree}{3\text{-}\textsc{SAT}}
\newcommand{\kclique}{k\times k\text{-}\textsc{Clique}} 
\newcommand{\rabingame}{\textsc{Rabin Game}} 
\newcommand{\Oh}{\mathcal{O}}

%% file: 0Intro.tex
\section{Introduction}

We study Rabin games defined as follows. The arena of a Rabin game is a (finite) directed graph $D$ whose vertices are divided among the two players involved: Steven and Audrey\footnote{The right way to memorize the player names is St{\bf{even}} and {\bf{Odd}}rey; the naming comes from the context of parity games.}. There is an initial vertex $u_1$ on which a token is initially placed. The game proceeds in turns. Each turn, the player controlling the vertex~$u$ on which the token is currently placed chooses any outneighbour $v$ of $u$ and moves the token from $u$ to~$v$. Thus, by moving the token, the players construct an infinite walk $\rho=(u_1,u_2,u_3,\ldots)$ in $D$, called a {\em{play}}. To determine the winner, the play $\rho$ is compared against the winning condition consisting of $k$ pairs of vertex subsets $(G_1,B_1),(G_2,B_2),\ldots,(G_k,B_k)$ as follows: Steven wins if there exists $i\in \{1,\ldots,k\}$ such that along $\rho$, $G_i$ is visited infinitely often while $B_i$ is visited only a finite number of times; Audrey wins otherwise. The computational question associated with the game is to determine which player has a winning strategy.

Rabin conditions were first introduced by Rabin in his proof of decidability of S2S (monadic second order with two successors)~\cite{Rab69}. They also naturally appear in the determinization of B\"uchi automata~\cite{Saf88,Pit06,Sch09}, a key step in the synthesis problem for reactive systems with specifications given in Linear Temporal Logic. Since then, algorithms for solving Rabin games have been extensively studied~\cite{EJ99,KV98,PP06,BMMSS22}. 
%\antonio{More context on Rabin conditions in comment}
%Rabin games  have been extensively studied in the area of formal verification in the connection with the synthesis problem for reactive systems; see~\cite{EJ99,PP06}. 
They generalise the more well-known {\em{parity games}}, which differ by altering the winning condition as follows. Each vertex of the graph bears a {\em{colour}}, which is an integer from $\{1,\ldots,k\}$.
Steven wins a play $\rho$ if the largest colour seen infinitely often in $\rho$ is even, and otherwise Audrey wins. Indeed, to reduce a parity game with colours $\{1,\ldots,k\}$ to a Rabin game with $\lfloor k/2\rfloor$ pairs in the winning condition, it suffices to take the same graph $D$ and set 
$$G_i=\{\textrm{vertices with colours }\geq 2i\}\qquad \textrm{ and }\qquad B_i=\{\textrm{vertices with colours }\geq 2i+1\},$$
for all $i\in \{1,\ldots,\lfloor k/2\rfloor\}$. Further, both parity games and Rabin games are generalised by {\em{Muller games}}, where again vertices have colours from $\{1,\ldots,k\}$ (each vertex may bear multiple colours), and the winning condition is defined by simply providing a family $\mathcal{F}$ of subsets of $\{1,\ldots,k\}$ that are winning for Steven in the following sense: Steven wins a play $\rho$ if the set of colours seen infinitely often in $\rho$ belongs to $\mathcal{F}$.

In a breakthrough paper, Calude, Jain, Khoussainov, Li, and Stephan~\cite{CJKLS22} proved that solving all the three games discussed above is fixed-parameter tractable when parameterised by $k$ (the number of colours, respectively the number of pairs in the winning condition).
More precisely, determining the winner of the game can be done in $k^{\Oh(k)}\cdot  n^{\Oh(1)}$ time, where $n$ is the number of vertices of the arena. The recent work of Majumdark, Sa\v{g}lam and Thejaswini~\cite{MST23} provides a more precise analysis which results in an algorithm solving Rabin games in polynomial space and time $k!^{1+o(1)}\cdot nm$, where $m$ is the number of~edges. While the work of Calude et al. also provided a quasipolynomial-time algorithm to solve parity games, it is known that solving Rabin games is already ${\mathsf{NP}}$-complete~\cite{EJ88,EJ99}, while solving Muller games is ${\mathsf{PSPACE}}$-complete~\cite{HD05}. Hence, for those games, the existence of (quasi)polynomial-time algorithms is~unlikely.

In their work, Calude et al.~\cite{CJKLS22} provided also complexity lower bounds based on the Exponential Time Hypothesis (ETH, the assumption that there exists $\delta>0$ such that {\sc{3SAT}} problem cannot be solved in time $\Oh(2^{\delta n})$) for some of the games discussed above. They proved that assuming ETH, there are no algorithms with running time $2^{o(k\log k)}\cdot n^{\Oh(1)}$ for solving Muller games with priorities in $\{1,\ldots,k\}$ or $d$-dimensional $k$-parity games (see preliminaries for a definition of the latter variant). Since every $k$-dimensional parity game can be reduced in polynomial time to a Rabin game with $k$ pairs in the winning condition (see~\cite{CHP07}), one can also derive, as a corollary, the same lower bound for solving Rabin games. The reduction provided by Calude et al. starts with the {\sc{Dominating Set}} problem and is rather involved.

\paragraph*{Our contribution.} We provide a simple reduction that reproves the tight complexity lower bound for solving Rabin games that follows from the work of Calude et al. More precisely, we prove that assuming ETH, there is no algorithm for this problem with running time $2^{o(k\log k)}\cdot n^{\Oh(1)}$. The same lower bound for (the more general) Muller games follows as a direct corollary. By a minor twist of our construction, we can also reprove the lower bound for $k$-dimensional parity games reported by Calude et al. 

We believe that our reduction is significantly simpler and more transparent than that of Calude et al. but more importantly, it gives a better insight into the origin of the $2^{o(k\log k)}$ factor in the complexity of the problem. Analyzing the algorithms of~\cite{CJKLS22,MST23}, this factor stems from considering all possible permutations of the $k$ pairs of vertex subsets involved in the winning condition. In our reduction, those permutations form the space of potential solutions of a carefully chosen pivot problem --- {\sc{Permutation SAT}}, a special case of a temporal constraint satisfaction problem --- which we discuss below.

\paragraph*{Temporal CSPs and Permutation SAT.} 
A constraint satisfaction problem (CSP) is the problem of deciding if there exists a variable assignment
that satisfies a given set of constraints. \emph{Temporal problems} is a rich family of CSPs that model planning
various events on a timeline.
In a basic form, every variable corresponds to an event that needs to be scheduled at some point of time 
and constraints speak about some events being in specific order (e.g., one preceding another), at the same time,
or at different times. 
This is usually modeled with $\mathbb{Q}$ as the domain and constraints having access to predicates $<$, $\leq$, $=$, and $\neq$. A $\mathsf{P}$ vs $\mathsf{NP}$ dichotomy for finite languages within this formalism has been provided
by Bodirsky and K\'{a}ra~\cite{BodirskyK10}.

An instance of such a temporal CSP with $k$ variables and $n$ constraints can be solved in time
$k^k \cdot (k+n)^{\Oh(1)}$ as follows: since variables are accessed only via comparisons $<$, $\leq$, $=$, and $\neq$, 
without loss of generality one can restrict to assignments with values in $\{1,2,\ldots,k\}$,
and there are $k^k$ such assignments that can be all checked. 
An interesting and challenging question is:
For which languages this running time can be significantly improved?

In this paper, we focus in a particular temporal CSP: \textsc{Permutation SAT}. An instance of this problem is given by a boolean combination of literals of the form $x_1 < x_2 < \ldots < x_{\alpha}$; a solution for it is an assignment of variables to integers making it a valid formula. We say that such a problem is an instance of $(\alpha,\beta)$-\textsc{Permutation SAT} if its constraints use at most $\beta$ literals, and each of these literals involves at most $\alpha$ variables.
Observe that, without loss of generality, in \textsc{Permutation SAT} one can restrict attention to assignments
being surjective functions from variables $\{x_1,\ldots,x_k\}$ to $\{1,\ldots,k\}$, which can be interpreted as permutations
of $\{1,\ldots,k\}$; this justifies the choice of the problem name and yields a brute-force algorithm with running time
$k! \cdot (k+n)^{\Oh(1)}$. 

Bonamy et al.~\cite{BonamyKNPSW18} proved that $(3,\infty)$-\textsc{Permutation SAT} admits 
no $2^{o(k \log k)} n^{\Oh(1)}$ algorithm unless the Exponential Time Hypothesis (ETH) fails. 
Our main technical contribution is a similar lower bound for $(2,4)$-\textsc{Permutation SAT} (Theorem~\ref{thm:permSAT-hardness}).
The proof of this result is a simple reduction from the {\sc{$k\times k$-Clique}} problem considered by Lokshtanov, Marx, and Saurabh~\cite{LMS18}. It is our belief that $(\alpha,\beta)$-{\sc{Permutation SAT}} is a problem with a very easy and robust formulation, hence its usefulness may extend beyond the application to Rabin games discussed in this work.

%Note that, without loss of generality, we can assume that $(\alpha,\beta)$-\textsc{Permutation SAT} asks for an assignment that is a surjective function from $\{x_1,x_2,\ldots,x_k\}$ to $\{1,2,\ldots,k\}$, which can be interpreted as a permutation of $[k]$; hence the name of the problem. This gives a brute-force algorithm with running time bound $k! \cdot (k+n)^{\Oh(1)}$.

%% file: 2PermSAT.tex
Fix integers $\alpha \geq 2$ and $\beta \geq 1$ and let $X$ be a finite set of variables. An \emph{$\alpha$-literal} is a predicate of the form
$x_1 < x_2 < \ldots < x_{\alpha'}$ (being a shorthand for 
$(x_1 < x_2) \wedge (x_2 < x_3) \wedge \ldots \wedge (x_{\alpha'-1} < x_{\alpha'})$) for some $2 \leq \alpha' \leq \alpha$
and variables $x_1,x_2,\ldots,x_{\alpha'}$ belonging to $X$; a \emph{literal} is a $2$-literal (i.e., a predicate of the form $x_1 < x_2$).
An \emph{$(\alpha,\beta)$-clause} is a disjunction of at most $\beta$ $\alpha$-literals, and an \emph{$(\alpha,\beta)$-formula} is a conjunction of $(\alpha,\beta)$-clauses. By $\beta$-clauses and $\beta$-formulas we mean $(2,\beta)$-clauses and $(2,\beta)$-formulas, respectively.

If $\phi$ is a formula with variable set $X$, then for a permutation $\pi$ of $X$ we define the satisfaction of (literals and clauses of) $\phi$ by $\pi$ in the obvious manner. In the $(\alpha,\beta)$-\textsc{Permutation SAT} problem we are given an $(\alpha,\beta)$-formula $\phi$ and the task is to decide whether there exists a permutation of the variables of $\phi$ that satisfies $\phi$.
$\beta$-\textsc{Permutation SAT} is a shorthand for $(2,\beta)$-\textsc{Permutation SAT}.

%and \textsc{Permutation SAT} is a shorthand for $(2,\infty)$-\textsc{Permutation SAT}, that is, without any restriction on the size of the disjunction in clauses. 

In this section we prove the following hardness result.

\begin{theorem}\label{thm:permSAT-hardness}
 Assuming ETH, there is no algorithm for $4$-\textsc{Permutation SAT} that would work in time $2^{o(k\log k)}\cdot n^{\Oh(1)}$, where $k$ is the number of variables and $n$ is the number of clauses.
\end{theorem}

%\begin{proof}[Proof of \cref{thm:permSAT-hardness}]
%\end{proof}

To prove Theorem~\ref{thm:permSAT-hardness} we use the problem $\kclique$ considered by  Lokshtanov, Marx, and Saurabh~\cite{LMS18}. They showed that, unless ETH fails, this problem cannot be solved in $2^{o(k\log k)}$-time.
We first define $\kclique$ below and then reduce $\kclique$ to $\permsat$.
    %Under the ETH, the work of Lokshtanov, Marx, and Saurabh~\cite{LMS18} showed that there is no  $O(2^{o(k\log k)})$-time solution to determine of an input is a positive instance of $\kclique$.

    An instance of the $\kclique$ problem is an undirected graph $\Gc$ with the vertex set $\{1,\ldots,k\}\times\{1,\ldots,k\}$ (which we can represent as a grid). This graph $\Gc$ is a positive instance of $\kclique$ if there is one vertex from each \emph{row} of the grid that forms a $k$-clique, that is, a $k$-clique in which no two vertices share the same first component.
    %An instance of $\kclique$ is an undirected graph $\Gc$ which  consists of a $k^2$ many vertices of a graph where they are arranged in the form of a $k\times k$-grid. This graph $\Gc$ is a positive instance of $\kclique$ if there is one vertex from each \emph{row} of the grid that forms a $k$-clique.  

\begin{theorem}[{\cite[Theorem~2.4]{LMS18}}]\label{thm:kxk-clique-hardness}
 Assuming ETH, there is no $2^{o(k\log k)}$-time algorithm for $\kclique$.
\end{theorem}

 \paragraph{The reduction.}  We now reduce $\kclique$ to $\permsat$. Suppose $\Gc$ is an instance of $\kclique$. We construct a $4$-formula $\phi_\Gc$ over variable set $X\coloneqq \{x_1,\dots,x_k,x_{k+1}, y_1,\dots,y_k\}$ as~follows.
 
%We construct the formula so that a permutation of the literals  is a positive instance if  $x_1<x_2<\dots<x_k<x_{k+1}$ and  $x_i<y_{i_1}<\dots<y_{i_\ell}<x_{i+1}$, if the selected $k$-clique contains vertices in column $i$ contains exactly elements in rows $i_1,\dots,i_\ell$. 

\begin{figure}[t]
    \centering
    \includegraphics[]{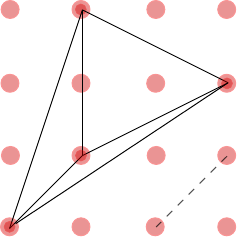}
%\makebox[\textwidth][c]{\input{onlytree.pdf_tex}}%
  \caption{The construction in Section~\ref{sec:permSAT}. The highlighted clique corresponds to permutation $x_1<y_4<x_2<y_1<y_3<x_3 < x_4<y_2<x_5$ (with $y_1$ and $y_3$ possibly swapped). The dashed non-edge $((4,3),(3,4))$ is disallowed by the clause $\neg\left((x_{4}< y_{3}<x_{5}) \land (x_3<y_4<x_{4})\right)$ which ensures if $y_4$ appears between $x_3$ and $x_4$, then $y_3$ does not appear between $x_4$ and $x_5$.
    \label{fig:kClique}}
\end{figure}

Recall that the vertices of the graph $\Gc$ are of the form $(i,j)$ for $i,j\in \{1,\dots, k\}$. 
We say that vertex $(i,j)$ is in the $i^{th}$ row and $j^{th}$ column. To construct $\phi_\Gc$, we first write the following $3k$ many $1$-clauses:
\begin{align*}
x_1<x_2,\quad x_2<x_3,\quad \ldots,\quad x_k< x_{k+1},\\
 x_1<y_1,\quad x_1<y_2,\quad \ldots,\quad x_1<y_k\\
 y_1<x_{k+1},\quad y_2<x_{k+1},\quad \ldots,\quad y_k<x_{k+1}
\end{align*}
The conjunction of these clauses ensures that in any permutation satisfying $\phi_{\Gc}$, the variables $x_1,\ldots,x_{k+1}$ are ordered exactly in this way, while variables $y_1,\ldots,y_k$ are sandwiched between $x_1$ and $x_{k+1}$. In other words, the $y$-variables that are placed between $x_j$ and $x_{j+1}$ indicate the rows that choose their clique vertices from the $j^{\text{th}}$ column; and for some j's, this set may be empty as well.

Next, we introduce clauses that restrict the placement of variables $y_1,\ldots,y_k$ within the chain $x_1<x_2<\ldots<x_{k+1}$. The intention is the following: placing $y_i$ between $x_j$ and $x_{j+1}$ corresponds to choosing the vertex $(i,j)$ to the clique. Hence, it remains to introduce clauses ensuring that vertices chosen in this way in consecutive rows are pairwise adjacent. To this end, for every pair $(a,b),(c,d)$ of vertices non-adjacent in $\Gc$, we construct the following $4$-clause:
%%Ant: Subscript corrected
$$( y_{a}<x_{b})\lor (x_{b+1}<y_a)\lor (y_c<x_d)\lor (x_{d+1}<y_c).$$
Note that logically, this $4$-clause is equivalent to the following:
$$\neg\left((x_{b}< y_{a}<x_{b+1}) \land (x_d<y_c<x_{d+1})\right).$$
Thus, intuitively speaking, the $4$-clause forbids simultaneously choosing $(a,b)$ and $(c,d)$ to the clique.

This concludes the construction of the formula $\phi_\Gc$. It remains to verify the correctness of the reduction.

 \begin{lemma}   
 The graph $\Gc$ admits a $k$-clique with one vertex from each row if and only if $\phi_\Gc$ is satisfiable.
 \end{lemma}
 \begin{proof}
 First, suppose $\Gc$ contains a $k$-clique $K = \{(1,b_1),\dots, (k,b_k)\}$.
 Consider any permutation $\pi$ of $X$ such that
 \begin{itemize}[nosep]
     \item $x_1<x_2<\dots<x_k<x_{k+1}$, and
     \item $x_{b_i}<y_i<x_{b_i+1}$, for all $j\in \{1,\ldots,k\}$.
 \end{itemize}
 (Note that $\pi$ is not defined uniquely, the relative placement of $y_i$ and $y_{i'}$ can be arbitrary whenever $b_i=b_{i'}$.) It can be easily seen that $K$ being a clique, implies that all clauses in $\phi_{\Gc}$ are satisfied. The $1$-clauses are satisfied trivially, while every $4$-clause constructed for a non-adjacent $(a,b),(c,d)$ is satisfied because $(a,b)$ and $(c,d)$ cannot simultaneously belong to~$K$.

 Suppose now that there is an ordering of $X$ that satisfies $\phi_\Gc$. Clearly, it must be the case that $x_1<x_2<\dots<x_k<x_{k+1}$. Further, for every $i\in \{1,\ldots,k\}$ we have $x_1<y_i<x_{k+1}$ and therefore, there exists $j_i$ such that $x_{j_i}<y_i<x_{j_i+1}$. We let $K \coloneqq \{(i,j_i) \colon i\in\{1,\ldots,k\}\}$; note that $K$ contains one vertex from each row.  We claim that $K$ is a clique in $\Gc$.  
    Indeed, since in $\phi_\Gc$ there is a clause disallowing that $\left((x_{b}< y_{a}<x_{b+1}) \land (x_d<y_c<x_{d+1})\right)$ whenever there is no edge between $(a,b)$ and $(c,d)$, all vertices of $K$ must be pairwise adjacent.    
\end{proof}

This concludes the proof of Theorem~\ref{thm:permSAT-hardness}.
We remark that establishing the complexity of $2$- and $3$-\textsc{Permutation SAT} remains an interesting and challenging
open problem. Eriksson in his MSc thesis~\cite{Eriksson} shows that $2$-\textsc{Permutation SAT} can be solved
in time $((k/2)!)^2 \cdot (k+n)^{\Oh(1)}$, which gives roughly  a $2^{k/2}$ multiplicative improvement over the naive algorithm.

For a broader context, we also remark that
a more general variant of \textsc{Permutation SAT} is \textsc{Permutation MaxSAT}, where we ask for an assignment
that satisfies as many constraints as possible (instead of asking to satisfy all of them). 
Observe that $(2,1)$-\textsc{Permutation SAT} is equivalent to a problem of checking if a given directed graph is acyclic
(and thus solvable in polynomial time) while $(2,1)$-\textsc{Permutation MaxSAT} is equivalent to finding a maximum
acyclic subdigraph (which is NP-hard).
A simple folklore dynamic programming algorithm solves $(2,1)$-\textsc{Permutation MaxSAT} in $2^{\Oh(k)} n^{\Oh(1)}$ time
and this algorithm can be generalised to $(3,1)$-\textsc{Permutation MaxSAT}~\cite{BodlaenderFKKT12}.
Kim and Gon\c{c}alves~\cite{KimG13} proved that $(4,1)$-\textsc{Permutation MaxSAT} 
admits no $2^{o(k \log k)} n^{\Oh(1)}$ algorithm unless the Exponential Time Hypothesis fails.

%% file: 3RabinGames.tex
Finally, in this section, we prove the main result of this paper, stated as Theorem~\ref{thm:rabinLB} below.

\begin{theorem}~\label{thm:rabinLB}
    Assuming the Exponential Time Hypothesis, there is no algorithm that solves Rabin games with $n$ vertices and degree $k$ in time $2^{o(k\log k)}\cdot n^{\Oh(1)}$.   
\end{theorem}

As mentioned earlier, we reduce from $\permsat$.

\paragraph{The reduction.}
    Let $\phi=C_1\land C_2\land \dots\land C_m$ be an instance of $\permsat$ over $k$ variables $\{y_1,\dots,y_k\}$, where $C_1,\ldots,C_m$ are $4$-clauses. We construct an instance of $\rabingame$ such that, in this instance, there is a strategy for \Steven~iff $\phi$ is satisfiable. 
    \begin{figure}[ht]%{l}{0.4\textwidth}
    \centering
    \includegraphics[]{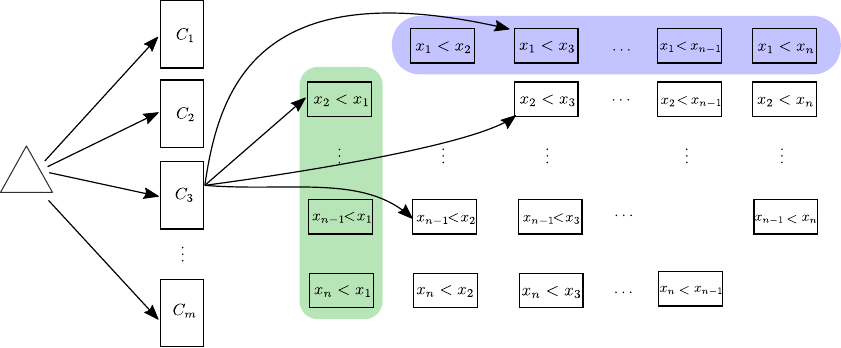}
%\makebox[\textwidth][c]{\input{onlytree.pdf_tex}}%
  \caption{Part of the constructed game graph $D$. The clause $C_3$ is $(x_1<x_3)\lor(x_2<x_1)\lor (x_2<x_3)\lor (x_{n-1}<x_2)$.
    Vertices of $G_1$ are highlighted in green and vertices of $B_1$ are highlighted in blue.
    \label{fig:psatRabin}}
\end{figure}

   We first define the game graph $D$; see Figure~\ref{fig:psatRabin}.
   There is an initial vertex $\Delta$, as well as vertices $[C_1],\dots,[C_m]$, one for each of the $m$ $4$-clauses in $\phi$. Further, for each possible literal $x_i<x_j$, where $i,j\in \{1,\ldots,k\}$ and $i\neq j$, there is a vertex $[x_i < x_j]$. Vertex $\Delta$ belongs to \Audrey, while all other vertices belong to \Steven.
   
   The intention is that whenever \Audrey{} moves the token currently placed at $\Delta$, she chooses a clause that she wishes to see satisfied. To facilitate this, we add edges 
   $\Delta\rightarrow [C_\ell]$ for all $\ell\in \{1,\ldots,m\}$. Once the token is at a vertex $[C_\ell]$, Steven needs to respond with a literal present in $C_i$; the intention is for it to be a true literal in $C_i$. Therefore, for every clause $C_\ell$ and literal $x_i<x_j$ present in $C_\ell$, we add the edge $[C_\ell] \rightarrow [x_i<x_j]$. Finally, to allow \Audrey{} checking further clauses, we add edges back to $\Delta$: for every literal $x_i<x_j$, there is an edge 
   $[x_i<x_j]\rightarrow \Delta$.

   Next, we define subset pairs constituting the winning condition.
  For each $i\in \{1,\ldots,k\}$, we set $$G_i = \{[x_j<x_i]\colon j\in\{1,\ldots,k\}\setminus \{i\}\}\qquad \text{and} \qquad B_i = \{[x_i<x_j]\colon j\in \{1,\ldots,k\}\setminus\{i\}\}.$$ 

  Before we proceed to the formal verification of the correctness of the reduction, let us give some intuition.
    It is easy to see that every third turn, the token is placed at vertex $\Delta$. At each such moment, turn Audrey chooses to move the token to any vertex corresponding to a clause $C_\ell$, with the intention of challenging Steven about the satisfaction of $C_\ell$. Then Steven has to declare the literal that satisfies $C_\ell$. If Steven tries to ``cheat'' by picking literals that cannot be extended to a full ordering of the variables, then the winning condition is designed in such a way that the play will be losing for him. 
    Consider the illustration in Figure~\ref{fig:psatRabin}, where for an instance $\phi$ of $\permsat$ which consists of $m$ clauses such that the clause $C_3$ is $(x_1<x_4)\lor(x_2<x_1)\lor (x_2<x_3)\lor (x_{n-1}<x_2)$.
    The vertices in $G_1$ are highlighted in green and the vertices in $B_1$ are highlighted in blue.

\begin{lemma}
        The instance $\phi$ of $\permsat$ is satisfiable if and only if Steven has a winning strategy in the constructed Rabin game. 
\end{lemma}
\begin{proof}
    First suppose $\phi$ is satisfiable, consider a satisfying permutation $\pi$. This gives rise to a (positional) winning strategy for \Steven: For each vertex $[C_\ell]$, \Steven~picks the edge leading to the vertex $[x_i<x_j]$ corresponding to any literal of $C_\ell$ that is satisfied in $\pi$. Consider now any infinite play $\rho$ where Steven obeys this strategy. Let $L$ be the set of literals visited infinitely often by $\rho$, and let $i_\mymax$ be such that $x_{i_\mymax}$ is the variable that is the largest in $\pi$ among variables appearing in the literals of $L$. We argue that $\rho$ satisfies the constructed Rabin condition with the index $i_\mymax$ as a witness. This is because   %for some $i$, and never visits $[x_{i_\mymax}<x_i]$ for any $i$. % and 
    %Thus 
     $L$ intersects $G_{i_\mymax}$ as $\rho$ visits  $[x_i<x_{i_\mymax}]$ infinitely often for some $i$, while the intersection of $L$ with $B_{i_\mymax}$ is empty, as $\rho$ never visits any vertex $[x_{i_\mymax}<x_i]$ for any $i$.
     
    Suppose now $\phi$ is not satisfiable. Then we need to show that \Audrey~can win against any positional strategy of \Steven.
    Indeed, consider a fixed positional strategy of Steven: 
    for each Steven vertex $[C_\ell]$ the strategy picks an edge $[C_\ell]\to [x_{a_\ell}<x_{b_\ell}]$ for some literal $x_{a_\ell}<x_{b_\ell}$ appearing in $C_\ell$.
    Since $\phi$ is not satisfiable, the set $\left\{x_{a_\ell}<x_{b_\ell} \colon \ell\in[m]\right\}$ of all selected literals has a cycle. That is, there are variables $x_{c_1},\dots,x_{c_p}$ such that literals $x_{c_1}<x_{c_2}, x_{c_2}<x_{c_3}, \ldots,x_{c_{p-1}}<x_{c_p}, x_{c_p}<x_{c_1}$ are among those selected by Steven's strategy. Observe now that for the fixed Steven's positional strategy, \Audrey{} may set up a counter strategy that repeatedly visits each of the vertices $[c_1<c_2],[c_2<c_3],\ldots,[c_{p-1}<c_p],[c_p<c_1]$ in a cycle, so that these are exactly the literal vertices visited infinitely often in the play. Then this play does not satisfy the constructed Rabin condition, since for each $i\in \{1,\ldots,k\}$, the set of vertices occurring infinitely often either intersects both $B_{i}$ and $G_{i}$ (if $i\in \{c_1,\ldots,c_p\}$), or is disjoint with both $B_i$ and $G_i$ (if $i\notin \{c_1,\ldots,c_p\}$). Hence, Audrey may win against any fixed positional strategy of Steven.
\end{proof}

Using the reductions shown in the preliminaries, we obtain similar corollaries for Muller and generalised parity objectives.

\begin{corollary}
Assuming the Exponential Time Hypothesis, there is no algorithm that solves Muller games with $n$ vertices and $k$ colours in time $2^{o(k\log k)}\cdot n^{\Oh(1)}$.
\end{corollary}

\begin{corollary}
Assuming the Exponential Time Hypothesis, there is no algorithm that solves $d$-dimensional $3$-parity games with $n$ vertices in time $2^{o(d\log d)}\cdot n^{\Oh(1)}$.
\end{corollary}

We conclude by remarking that we can also extend our result to $2$-dimensional $k$-parity games. Indeed, consider the following assignment of colours to the same game graph $D$: for each vertex of the form $[x_j<x_i]$, we assign the two-dimensional colour $(2j+1, 2i)$. The correctness of this reduction is similar to that for Rabin games presented above, hence we leave the verification to the reader.
\begin{corollary}
Assuming the Exponential Time Hypothesis, there is no algorithm that solves $2$-dimensional $k$-parity games with $n$ vertices in time $2^{o(k\log k)}\cdot n^{\Oh(1)}$.
\end{corollary}